\documentclass[DIV=classic,a4paper,10pt]{myart}

\KOMAoptions{DIV=last}
\usepackage{atbegshi}
\AtBeginDocument{\AtBeginShipoutNext{\AtBeginShipoutDiscard}}

\usepackage{caption}
\usepackage{subcaption}

\newcommand{\abs}[1]{\left\vert#1\right\vert}

\renewcommand{\mid}{\,|\,}

\usepackage{accents}
\newcommand{\ubar}[1]{\underaccent{\bar}{#1}}

\begin{document}

\title{Efficient hedging under ambiguity in continuous time}\thanks{Heartfelt thanks are due to Daniel Bartl for fruitful discussions.}

\author[]{Ludovic Tangpi}

\abstract{
It is well known that the minimal superhedging price of a contingent claim is too high for practical use. %
In a continuous-time model uncertainty framework, we consider a relaxed hedging criterion based on acceptable shortfall risks.
Combining existing aggregation and convex dual representation theorems, we derive duality results for the minimal price on the set of upper semicontinuous discounted claims.
}

\date{\today}
\keyAMSClassification{91B30, 91G80, 60H30, 60G48.}
\keyWords{Superhedging, model ambiguity, acceptance set, risk measure, optimized certainty equivalent, volatility uncertainty.}
\maketitleludo
\setcounter{page}{1} %

\section{Introduction}
In this paper, we are concerned with convex duality for the minimal superhedging problem with non-zero shortfall risk, in continuous time.
In a financial market with underlying $S$, the minimal superhedging price $\phi(X)$ of a discounted contingent claim $X$ is the smallest cost $m$ needed to form a superhedging portfolio.
That is, to find an admissible strategy $Z$ such that
\begin{equation}
\label{eq:hedge}
	m + (Z\cdot S)_T \ge X,
\end{equation}
where $(Z\cdot S)_T$ is the total gain up to time $T \in (0,\infty)$ from trading $S$. 
A classical result at the heart of mathematical finance gives conditions guaranteeing a pricing-hedging duality, i.e. ensuring that $\phi(X)$ is the largest non-arbitrage price of $X$, see e.g. \citet{DS94} for details, and \citet{Kra-Sch} for applications to portfolio optimization.

In the presence of model ambiguity, i.e. when the negligible events do not stem from a single measure, the pricing hedging-duality has attracted a sustained of attention.
Notably under models with volatility uncertainty, such duality results have been derived e.g. by \citet{peng_2010}, \citet{denis06} and \citet{STZ1,STZ2} for contingent claims that are (versions of) continuous random variables.
A crucial step to derive most of these results is to prove a dynamic programming principle. 
This requires the (``dynamic version'' of the functional) $\phi$ to be time-consistent.
\citet{Neu-Nutz} have extended these representations (and the dynamic programming principles) to measurable claims using the theory of analytic sets.
In the model-free framework, i.e. when no probabilistic assumption is made, superhedging duality results include those in \cite{Acciaio2016,bei-hl-pen,robhedging,Bur-Fri-Mag17,Burz16} in discrete time and \cite{Dolinsky2014,Hou2015,liminf,Bar-Neu-Kup17} in continuous time.

It is well-known that the minimal superhedging price is too high for practical use, and even higher under model uncertainty.
This motivated the notion of quantile-hedging introduced by \citet{Foel-Leu,Foel-Leu2000} and further developed into risk-based approaches by \citet{Arai10} and \citet{Rud07}; and into no-good deal based valuations by \citet{Nadal}.
We also refer to \citet{Bec-Ken17} for the analysis of robust no-good deals.
More precisely, this consists into substituting the (strict) superhedging requirement \eqref{eq:hedge} by the relaxed condition

\begin{equation}
\label{eq:risk-hedge}
	m + (Z\cdot S)_T-X \in {\cal A},
\end{equation}
where ${\cal A}$ is the acceptance set of a convex monetary risk measure, or a set of acceptable discounted financial positions.
Adjusting the set ${\cal A}$ allows to change the level of risk aversion.
Under model uncertainty, superhedging dualities in such a situation have been investigated by \citet{robhedging} in the discrete time model-free framework.

The goal of the present paper is to investigate the continuous-time case when a set ${\cal P}$ of possible reference measures on the canonical space $C([0,T],\mathbb{R}^d)$ is fixed.
Notice that if the risk measure with acceptance set ${\cal A}$ is not time-consistent, then the resulting superhedging functional is not necessarily time-consistent anymore, rendering the dynamic programming approach prevalent in the literature harder to apply.
The proposed argument is based on results by \citet{Robdual} which give conditions under which a continuity from below condition (also known as Fatou property) yields a representation of convex monotone functions.
More precisely, we show that a suitable sequential closedness of the acceptance set ${\cal A}$ carries over to the sublevel sets of the superheding functional, guaranteeing enough regularity to derive a convex dual representation; see Theorem \ref{thm:main} for a precise statement.
This will require the use of aggregation results developed by \citet{STZ10}.
As application, the case where the shortfall risk is quantified by a robust optimized certainty equivalent is systematically studied.
Recall that this risk measure is not time-consistent, except in the case where it corresponds to the entropic risk measure.

In the next section, we precise the probabilistic setting and the main results of the paper.
Namely, a convex dual representation for the superhedging functional for upper semicontinuous claims when the shortfall risk is quantified by a risk measure whose acceptance set satisfies some integrability property.
As example, the case of robust optimized certainty equivalent is studied in details, since this class of risk measures includes a large number of example, see e.g. \cite{Roboptim,Ben-Teb}.
All the proofs are given in Section \ref{sec:proofs}, and an appendix contains some technical concepts from \cite{STZ10}.

\section{Setting and main results}
\subsection{Probabilistic setting}

The findings of this work rely on a representation results of \citet{Robdual} and the aggregation results of \citet{STZ10} from which we borrow the probabi{}listic setting.
More precisely, fix $T\in (0,\infty), d \in \mathbb{N}\setminus\{0\}$ and let $\Omega$ be the canonical space of $\mathbb{R}^d$-valued continuous paths on $[0,T]$ with $\omega_0=0$.
Let $P_0$ be the Wiener measure on $\Omega$ and $S$ the canonical process, with natural filtration $\mathbb{F}^S=({\cal F}^S_t)_{t\in [0,T]}$.
By \citet{karandikar83},
there exists an $\mathbb {F}^S$-adapted, continuous 
process $\langle S \rangle$, such that $\langle S \rangle_t=\langle S \rangle_t^Q$ $Q$-a.s. 
for all $t\in[0,T]$, and every local martingale measure $Q$ of $S$, where $\langle S\rangle^Q$ denotes the $Q$-quadratic variation of $S$. 
Let $\hat{a}$ be the density of the quadratic variation $\langle S\rangle$ given by
\begin{equation*}
	\hat{a}_t := \limsup_{\varepsilon \downarrow 0}\frac{1}{\varepsilon}\left(\langle S\rangle_t - \langle S\rangle_{t-\varepsilon} \right).
\end{equation*}
 We denote by ${\cal M}(S)$ the set of all local martingale measures $P$ such that $P$-a.s., $\langle S\rangle_t$ is absolutely continuous in $t$ and $\hat{a}$ is valued in the set $\mathbb{S}_d^{>0}$ of symmetric positive definite matrices.
 For every $P \in {\cal M}(S)$ and every integrable $\mathbb{F}^S$-progressively measurable process $a$ taking values in $\mathbb{S}_d^{>0}$, and such that $a = \hat {a}$ $P$-a.s. (such process $a$ is called diffusion coefficient),
$P$ is a weak solution of the SDE
 \begin{equation}
 \label{eq:sde}
 	dY_t = a^{1/2}(Y_\cdot)dS_t \quad P_0\text{-a.s.} 
 \end{equation}
 with initial value $P(S_0=0)= 1$.
 In particular, $S$ is a $P$-local martingale.
 By \cite{karatzas01}, the SDE \eqref{eq:sde} admits a unique weak solution for every bounded process with values in $ \mathbb{S}_d^{>0}$.
Let $A_0$ be a generating class of diffusion coefficients (see Definition \ref{def:generating}) such that every $a \in A_0$ is bounded and $P^a$ satisfies the martingale representation property.
Further let  $A$ be a separable class of diffusion coefficients generated by $A_0$, see Definition \ref{def:separable}, and put 
$${\cal P}: = \{P^a: a \in A\}.$$
We consider the set ${\cal P}$ as the set of reference probability measures. 
For every $P \in {\cal P}$, let $\mathbb{F}^P:=({\cal F}_t^{P})_{t\in [0,T]}$ 
be the $P$-completion of the right continuous version 
of the filtration $\mathbb{F}^S$, and denote by $\mathbb{F}: =({\cal F}_t)_{t\in [0,T]}$ the universal filtration given by ${\cal F}_t:= \bigcap_{P \in {\cal P}}({\cal F}^P_t\vee{\cal N}^{\cal P})$, where ${\cal N}^{\cal P}$ is the collection of $P$-null sets for all $P \in {\cal P}$.

Let $L^0({\cal P})$ be the space of ${\cal F}_T$-measurable random variables which are identified if they agree ${\cal P}$-q.s.\footnote{Hereby ${\cal P}$-q.s. means $P$-a.s. for every $P \in {\cal P}$. Unless otherwise stated, all equalities and inequalities between random variables will be understood in this sense.} and, given $p \in [1, \infty)$, we denote by $L^p({\cal P})$ the space of random variables $X \in L^0({\cal P})$ such that $E_P[\abs{X}^p]<\infty$ for all $P \in {\cal P}$.
Further let $L^\infty({\cal P})$ be the subspace of $L^p({\cal P})$ equipped with the norm $||X||_{\infty}:= \inf\{m>0: \sup_{P \in {\cal P}}P(|X|>m)=0\}$.

\subsection{Main results}
\label{sec:main results}
For every progressively measurable $\mathbb{R}^d$-valued processes $Z$ such that $\int_0^T\abs{Z_t}^2\,d\langle S\rangle_t < \infty$, we denote by $\int Z\,dS$ the usual It\^o's integral which implicitly depends on $P\in {\cal P}$ and by $M^Z$ the ${\cal P}$-q.s. unique $\mathbb{F}$-progressively measurable process such that $M^Z = \int Z\,dS$ $P$-a.s. for all $P \in {\cal P}$ (see \cite[Theorem 6.4]{STZ10}).
This defines a ${\cal P}$-local martingale, that is, a $P$-local martingale under each $P \in {\cal P}$.
The (admissible) gains and losses from trading in the financial market modeled by $S$ are given by the set
\begin{equation*}
 	G := \left\{M^Z_T: \int_0^t Z_u\,dS_u\ge -c \text{ for all } t \in [0,T], \text{ for some } c>0 \right\}
\end{equation*}
and the minimal superhedging cost $\phi(X)$ of a contingent claim $X \in L^0({\cal P})$, is given by
\begin{equation}
 \label{eq:superhedging}
 	\phi(X) := \inf\{m\in \mathbb{R}: m+ Y \ge X \text{ for some } Y\in G\}.
\end{equation}
Fix a non-empty, convex set ${\cal A} \subseteq L^0({\cal P})$ and containing $L^0_+({\cal P})$ assumed to be monotone\footnote{i.e. for every $X, X' \in L^1({\cal P})$ with $X\ge X'$ and $X' \in {\cal A}$, we have $X \in {\cal A}$.}.
The functional $\psi$ given by
\begin{equation}
\label{eq:defpsi}
 	\psi(X):= \inf\{m\in\mathbb{R}: m + Y -X \in {\cal A}\text{ for some }  Y \in G\}
\end{equation}
defines the minimal cost to be paid to construct a portfolio with a shortfall that lies in the set ${\cal A}$ but that may fail to superhedge the claim $X$ (in the ${\cal P}$-q.s. sense), with the convention $\inf \emptyset:=+\infty$.
Our aim is to derive the dual representation of the functional $\psi$.

Let $C_b$ and $U_b$ be the space of bounded continuous functions and bounded upper semicontinuous functions on $\Omega$, and $C_b({\cal P})$ and $U_b({\cal P})$ the set of elements of $L^\infty({\cal P})$ with a ${\cal P}$-q.s. version in $C_b$ and $U_b$, respectively.
Put
\[
\psi^\ast(Q):=\sup_{X \in C_b}(E_Q[X]-\psi(X))\quad \text{and} \quad \phi^\ast(Q):=\sup_{X \in C_b}(E_Q[X]-\phi(X))
\]
and denote by $\rho_{\cal A}$ and $\rho^*_{\cal A}$ respectively, the risk measure  associated to ${\cal A}$ and its conjugate, i.e.
\[
	\rho_{\cal A}(X):=\inf\{m: m+ X \in {\cal A}\} \quad \text{and}\quad \rho_{\cal A}^\ast(Q):=\sup_{X \in C_b}(E_Q[X]-\rho_{\cal A}(X)).	
\]
Further define
 $ca^+_1$, the set of Borel probability measures on $\Omega$, 
$ca^+_1(\tilde \Omega)$ its subset containing probability measures with support included in $\tilde{\Omega}:= \text{supp}({\cal P})$, the support\footnote{Here,
$\text{supp}({\cal P})$ is the unique closed set $\tilde{\Omega}$ for which $P(\tilde{\Omega}^c) = 0$ for all $P \in {\cal P}$ and 
$P(\tilde{\Omega}\cap O)>0$ for some $P \in {\cal P}$ whenever $O$ is open and $\tilde{\Omega}\cap O \neq \emptyset$. It can be checked that $\text{supp}({\cal P}) = \cup_{P \in {\cal P}}\text{supp}(P)$, where $\text{supp}(P)$ exists, since $P$ is a regular measure.} of ${\cal P}$ and by ${\cal M}^{\cal A}_{{\cal P}}(S)$ the set of probability measures $Q \in ca^+_1(\tilde \Omega)$ such that $S$ is a $Q$-local martingale and it holds $\rho^*_{\cal A}(-Q)<\infty$.
\begin{theorem}
\label{thm:main}
	Assume that the set ${\cal A}^-:= \{A^-: A\in {\cal A}\}$ is ${\cal P}$-uniformly integrable\footnote{i.e. it is $P$-uniformly integrable for each $P \in {\cal P}$; and $A^-:=\max(0, -A)$.}; the sublevel sets $\{ Q\in ca^+_1:\rho^*_{\cal A}(-Q)\le c\}$, $c\ge 0$ are weakly compact and $\liminf_{n \to \infty}A^n \in {\cal A}$ for every sequence $(A^n)$ in $ {\cal A}$ that is bounded in $L^1({\cal P})$.
	Then, if $\psi(0) > -\infty$, the functional $\psi$ is real-valued on $L^\infty({\cal P})$ and satisfies the representation\footnote{In \eqref{eq:representation_psi}, $E_Q[X]$ is understood as $E_Q[X']$, for any $X' \in C_b$ with $X = X'$ ${\cal P}$-q.s.
    This expectation is uniquely defined, see \cite[Lemma 4.5.1]{diss}.}
	\begin{equation}
	\label{eq:representation_psi}
	 	\psi(X) = \sup_{Q \in {\cal M}^{\cal A}_{{\cal P}}(S)}\left(E_Q[X] - \psi^\ast(Q) \right), \quad X \in C_b({\cal P}) 
	\end{equation}
	with $\psi^\ast(Q) = \phi^\ast(Q) + \rho^\ast_{{\cal A}}(-Q)$, for all $Q \in ca^+_1(\tilde \Omega)$.

	Moreover, if
	\[
		\psi^\ast(Q)=\sup_{X \in U_b}(E_Q[X]-\psi(X)),
	\]
	then one has 
	\begin{equation}
	\label{eq:representation_psi_Ub}
	 	\psi(X) = \sup_{Q \in {\cal M}^{\cal A}_{{\cal P}}(S)}\left(E_Q[X] - \psi^\ast(Q) \right), \quad X \in U_b({\cal P}).
	\end{equation}
\end{theorem}
The proof of this result is given in Subsection \ref{sec:proof_thm:main}.
This result is close in spirit to the so-called no-good deal bounds derived by \cite{Bec-Ken17} using the second order backward stochastic differential equations.
\begin{remark}
Since $G$ is convex, the condition ${\cal A}$ convex and monotone ensures that $\psi$ is increasing and convex on the vector space $L^1({\cal P})$. 
When ${\cal A}=L^0_+({\cal P})$, then $\psi$ reduces to the superhedging function $\phi$.
In Theorem \ref{thm:main}, assuming $\psi(0)>-\infty$ can be seen as a market viability condition, it is satisfied for instance if $\mathbb{R}_+\cap (G-{\cal A})  = \{0\}$, compare \cite{robhedging}, or if there is a probability measure $Q$ such that $E_Q[A]\ge 0$ for all $A \in {\cal A}$ and $E_Q[Y]\le 0$ for all $Y  \in G$.
Moreover, the condition ${\cal A}^-$ uniformly integrable prevents, in particular, that $\rho_{{\cal A}}$ attains the value $-\infty$ which is undesirable for a risk measure.
Furthermore, assuming that $\liminf_{n \to \infty}X^n \in {\cal A}$ for every sequence $(X^n)$ in $ {\cal A}$ that is bounded in $L^1({\cal P})$ can be seen as a version of the Fatou's property for risk measures on $L^p$-spaces, see e.g. \cite{Kaina-Ruesch}.
\end{remark}

A particularly interesting case arises when ${\cal A}$ is the acceptance set of a robust optimized certainty equivalent.
More precisely, let $l\colon\mathbb{R}\to\mathbb{R}$ be a loss function satisfying the  usual assumptions
\begin{align}
\label{eq:ass.l}\tag{CIB}
\begin{array}{l}
l \text{ is convex, increasing, bounded from below, and } \\
l(0)=0, \,l^\ast(1)=0,\text{ and } l(x)>x \text{ for $|x|$ large enough}
\end{array}\bigg\}
\end{align}
where $l^*$ denotes the convex conjugate of $l$ defined as
\[ l^\ast(y):=\sup_{x\in\mathbb{R}} (xy-l(x)) \]
for $y\in \mathbb{R}$  and $l^*(+\infty):=+\infty$.
The functional $\rho:L^1({\cal P})\to (-\infty,\infty]$ defined by
\begin{equation}
\label{eq:rob.oce}
	\rho(X) := \inf_{m\in \mathbb{R}}\sup_{P \in {\cal P}}(E_P[l(m-X )]-m)
\end{equation}
 is the analogue, in the context of model ambiguity, of the optimized certainty equivalent risk measure introduced by \cite{Ben-Teb}.
 It satisfies 
 \begin{equation}
 \label{eq:rep_OCE}
  	\rho(X) = \sup_{Q \in ca^+_1}\left(E_Q[-X] - \inf_{P \in {\cal P}}E_P\left[l^*\left(\frac{dQ}{dP}\right) \right] \right),\quad X \in L^\infty({\cal P}),
  \end{equation}
  where $dQ/dP:=\infty$ if $Q$ is not absolutely continuous w.r.t. $P$ and with the understanding $E_P[Z]:=\infty$ whenever $E_P[Z^+]=\infty$, see \cite{Roboptim} for details.
Let us consider the acceptance set
\begin{equation*}
	{\cal A}:= \left\{X \in L^1({\cal P}): \rho(X)\le 0\right\}
\end{equation*}
and  denote by ${\cal M}^l_{{\cal P}}(S)$ the set of probability measures $Q \in ca^+_1(\tilde \Omega)$ such that $S$ is a $Q$-local martingale and it holds $\inf_{P \in {\cal P}}E_P[l^*(dQ/ dP) ]<\infty$.
\begin{theorem}
\label{thm:psiROCE}
	Assume that $l$ satisfies (CIB), there exist $a,b\ge 0$ and $p>2$ such that $l(x) \ge a|x|^p + b$ and $\psi(0)>-\infty$.
	If ${\cal P}$ is weakly compact, then, it holds
	\begin{align}
	\label{eq:psiROCE_Ub}
		\psi(X) &=\sup_{Q \in {\cal M}^l_{\cal P}(S) }(E_Q [X] - \inf_{P\in {\cal P}}E_P[l^*(dQ/dP)]),\quad X \in U_b({\cal P}).
	\end{align}
\end{theorem}
The proof of this result is given in Subsection \ref{sec:proof_thm:psiROCE}.
\begin{example}
	Let $\ubar{a},\bar{a}$ be two matrices with $0<\ubar{a}\le \bar{a}$, and denote by $A_0$ the set of (deterministic) functions on $[0,T]$ valued in $\mathbb{S}^{>0}_d$ and such that $\ubar{a}\le a_t\le \bar{a}$ for all $t\in [0,T]$. 
	By \cite[Example 4.9]{STZ10} $A_0$ is a generating class of diffusion coefficients and it clearly generates itself (in the sense of Definition \ref{def:separable}).
	Put 
	$$
		{\cal P} := \{ P: \langle S \rangle^P_t/dt \in A_0 \,\, \text{-$P\otimes dt$ a.s.}\}.
	$$
	It follows from \cite[Proposition 6.2]{Bion-Kervarec} that the set
	${\cal P}$ is compact in the weak topology.
	In this setting, taking for instance $l(x) = (x^+)^p/p$ with $p>2$, it follows from Theorem \ref{thm:psiROCE} that $\psi$ satisfies the representation
	\begin{equation*}
 		\psi(X) = \sup_{Q \in {\cal M}^l_{\cal P}(S)}\left(E_Q[X] - \inf_{P \in {\cal P}}\frac{1}{q}E_P\left[\left(\frac{dQ}{dP} \right)^q \right]\right),\quad X \in U_b({\cal P})
 	\end{equation*} 
 	with $q$ the H\"older conjugate of $p$. 
 \end{example}

\section{Proofs}
\label{sec:proofs}

\subsection{Proof of Theorem \ref{thm:main}}
\label{sec:proof_thm:main}
For each $c >0$, consider the set $G^c:=\{M^Z_T:\int_0^tZ_u\,dS_u\ge-c\text{ for all } t\in [0,T] \}$ and the functional
\begin{equation}
\label{eq:defpsi_c}
 	\psi^c(X):= \inf\{m\in\mathbb{R}: m + Y -X \in {\cal A}\text{ for some }  Y \in G^c\}.
\end{equation}
Recall that the inf-convolution $\rho_1\Box \rho_2$ of two functions $\rho_1$ and $\rho_2$ on $L^1({\cal P})$ is defined by
\begin{equation*}
	\rho_1\Box \rho_2(X) := \inf_{Y \in L^1({\cal P})}(\rho_1(X-Y) + \rho_2(Y)).
\end{equation*}
The minimal cost $\psi^c$ can be written as an inf-convolution:
\begin{lemma}
\label{lem:convol}
	For every $X \in L^\infty({\cal P})$, the minimal cost $\psi^c$ satisfies
		$\psi^c(X) = \rho_{\cal A} \Box\bar{\phi}^c (-X)$,
	whereby $\bar{\phi}^c(\cdot):=\phi^c(-\cdot)$ and $\phi^c(X):=\inf\{m\in \mathbb{R}: m+ Y - X\ge 0\text{ for some } Y \in G^c\}$.
\end{lemma}
\begin{proof}
	Let $X \in L^\infty({\cal P})$, $\varepsilon >0$ and $Y \in G^c$.
	There is $m \in \mathbb{R}$ such that $\rho_{\cal A}(Y-X) \ge m - \varepsilon$ and $m +Y-X \in {\cal A}$.
	Hence, $\psi^c(X)\le m\le  \rho_{\cal A}(Y-X) + \varepsilon$.
	This implies that $\psi^c(X) \le \inf_{Y \in G^c}\rho_{\cal A}(Y-X)$.
	On the other hand, if $\inf_{Y \in G^c}\rho_{\cal A}(Y-X) = -\infty$, the previous inequality is an equality.
	If $\inf_{Y \in G^c}\rho_{\cal A}(Y-X) > -\infty$, let $m \in \mathbb{R}$ be such that $m < \inf_{Y \in G^c}\rho_{\cal A}(Y-X)$. 
	Then it holds $m \le \psi^c(X)$ because if not, there would exist $Y \in G^c$ such that $m + Y-X \in {\cal A}$. 
	That is, $m \ge \rho_{\cal A}(Y-X)$.
    Therefore, 
    \begin{equation}
    \label{eq:conv_GA}
		\psi^c(X) = \inf_{Y \in G^c}\rho_{\cal A}(Y-X).
	\end{equation}
	In particular, denoting by $\bar{\phi}^c$ the functional $\bar{\phi}^c(X): = \inf\{m: m+X \in L^0_+({\cal P}) - G^c \}$, we have
    \begin{equation*}
	   	\psi^c(X) \ge \inf_{Y \in L^1({\cal P}) }(\rho_{\cal A}(-Y - X) + \bar{\phi^c}(Y))
	\end{equation*}
	and if we take $m > \inf_{Y \in L^1({\cal P}) }(\rho_{\cal A}(-Y - X) + \bar{\phi}^c(Y)) $, then for every $\varepsilon>0$ there exists $Y' \in L^1({\cal P})$ such that $m > \rho_{\cal A}(-Y'-X)+ \bar{\phi}^c(Y')-\varepsilon $.
	Thus, using definition of $\bar{\phi}^c$, one can find $Y \in G^c$ such that $\bar{\phi}^c(Y') + Y\ge -Y'-\varepsilon$.
	Since $\rho_{\cal A}$ is decreasing and translation invariant, this yields $m > \rho_{\cal A}(Y-X) - 2\varepsilon$, thus $m \ge \inf_{Y \in G^c}\rho_{\cal A}(Y-X)$ so that
	\begin{equation*}
	    \psi^c(X) = \inf_{Y \in L^1({\cal P}) }(\rho_{\cal A}(-Y - X+ \bar{\phi}^c(Y)) = \rho_{\cal A}\Box \bar{\phi}^c (-X).
	\end{equation*}
\end{proof}
\begin{proposition}
\label{pro:rep-c}
	Under the conditions of Theorem \ref{thm:main}, it holds:
	\begin{itemize}
		\item[(i)] For every claim $X\in L^0({\cal P})$ with $X^-\in L^1({\cal P})$ and $\psi^c(X) < \infty$, there exists an optimal $\bar{Y} \in G^c$ such that $\psi^c(X) + \bar{Y} - X \in {\cal A}$.
		\item[(ii)] The functional $\psi^c$ is real-valued on $L^\infty({\cal P})$ and satisfies
	\begin{align}
	\label{eq:rep_psi_c}
	 	&\psi^c(X) = \sup_{Q \in {\cal M}^{\cal A}(S)}\left(E_Q[X] - \psi^{c,\ast}(Q) \right), \quad X \in C_b({\cal P}) \\
	\label{eq:rep_psi_c_Ub}
	 	&\psi^c(X) \le \sup_{Q \in {\cal M}^{\cal A}(S)}\left(E_Q[X] - \psi^{c,\ast}(Q) \right), \quad X \in U_b({\cal P})
	\end{align}
	with ${\cal M}^{\cal A}(S)$ the set of probability measures $Q \in ca^+_1(\tilde \Omega)$ such that $\rho^*_{\cal A}(-Q)<\infty$; and it holds
	$$\psi^{c,\ast}(Q) = \phi^{c,\ast}(Q) + \rho^\ast_{{\cal A}}(-Q),$$for all $Q \in ca^+_1(\tilde{\Omega})$.
	\end{itemize}	
\end{proposition}
\begin{proof}
	(i) \emph{Existence:}	Let $X \in L^0({\cal P})$ with $X^-\in L^1({\cal P})$ and $\psi^c(X)< \infty$ be fixed.
	Let $(m^n)$ in $\mathbb{R}$ be a minimizing sequence satisfying $m^n \downarrow \psi^c(X)$ and for all $n\in \mathbb{N}$ there exists $Y^n \in G^c$ such that
	\begin{equation*}
		m^n + Y^n - X \in {\cal A}, \quad \text{with}\quad Y^n := \int_0^TZ^n_u\,dS_u.
	\end{equation*}
	 Let $M^{Z^n}$ be the unique process such that $M_t^{Z^n}=\int_0^tZ^n_u\,dS_u$ ${\cal P}$-q.s.
	It can be checked that $M^{Z^n}$ is a $P$-supermartingale for each $n \in \mathbb{N}$ and $P\in {\cal P}$.
	There exists a sequence $(A^n)$ in $ {\cal A}$ such that for every $n$, it holds $m^n + Y^n - X = A^n$.
	Since ${\cal A}^-$ is ${\cal P}$-uniformly integrable, $(A^n)^-$ is bounded in $L^1({\cal P})$ and $E_P[(A^n)^+] = E_P[A^n] + E_P[(A^n)^-]\le E_P[m^n + Y^n + X^-] + E_P[(A^n)^-]\le m^n + E_P[X^-]+ E_P[(A^n)^-]$.
	This shows that $(A^n)$ is bounded in $L^1({\cal P})$.
	Let $t \in [0,T]$ and put 
	$$Y_t:= \liminf_{n \to \infty}M^{Z^n}_t.$$
	The process $Y$ is $\mathbb{F}$-progressively measurable  and does not depend on a particular measure $P \in {\cal P}$.
	Since $\int_0^tZ^n_s\,dS_s\ge -c$, it holds $Y_t\ge -c$.
	On the other hand, it follows from Fatou's lemma and the $P$-supermartingale property of $\int Z^n\,dS$ that
	\begin{multline*}\textstyle
		E_P[Y^+_t] \le \liminf_{n\to \infty}E_P\left[\left(\int_0^tZ^n_u\,dS_u\right)^+\right]\\
				   \textstyle = \liminf_{n\to \infty} \left\{ E_P\left[\int_0^tZ^n_u\,dS_u\right] + E_P\left[\left(\int_0^tZ^n_u\,dS_u\right)^-\right]\right\}\le c.
	\end{multline*}
	That is, $Y_t \in L^1(P)$ and the process $Y$ is a $P$-supermartingale since for all $0\le s\le t\le T$ we have $Y_s = \liminf_{n \to \infty}M_s^{Z^n} \ge \liminf_{n \to \infty}E_P[M_t^{Z^n}\mid {\cal F}_s]\ge E_P[Y_t\mid {\cal F}_s]$.
	Let 
	$$\bar{Y}_t:=\limsup_{s\downarrow t,s \in \mathbb{Q}\cap [0,T]}Y_s \quad \text{for $t\in [0,T)$ and } \bar{Y}_T:=Y_T.$$
	Since the filtration $({\cal F}^P_t)_{t \in [0,T]}$ is right continuous for each $P \in {\cal P}$, the process $\bar{Y}$ is a c\`adl\`ag $P$-supermartingale with respect to $({\cal F}^P_t)_{t \in [0,T]}$, see \cite[Theorems VI.2 and VI.3]{Dellacherie1982}.
	Hence, $\bar{Y}$ is a $P$-supermartingale with respect to the filtration $\mathbb{F}$ for all $P\in {\cal P}$.
	Due	to \cite[Theorem 6.5 and Proposition 6.6]{STZ10}  there exist a $\mathbb{F}$-progressively measurable process $\bar{Z}$ and an increasing progressively measurable process $\bar{L}$ such that $\bar{L}_0 = 0$ and $\bar{Y}_t = \bar{Y}_0 + \int_0^t\bar{Z}_u\,dS_u - \bar{L}_t$, where $\int \bar{Z}\,dS$ is a ${\cal P}$-local martingale.
	Thus, $\int_0^t \bar{Z}_u\,dS_u\ge \bar{Y}_t-\bar{Y}_0\ge -c - \bar{Y}_0$ and by \cite[Theorem VI.2]{Dellacherie1982} and the right-continuity of our filtration, it holds $\bar{Y}_0 \le 0$ so that $M^{\bar{Z}}_T \in G^c$.
	Since $m^n + Y^n - X\in {\cal A}$ for all $n\in \mathbb{N}$ and ${\cal A}$ is monotone one has $\psi^c(X) + \bar{Y}_T - X \in {\cal A}$ and by
	$\psi^c(X) + \int_0^T\bar{Z}_u\,dS_u - X \ge \psi^c(X) +\bar{Y}_T - X$, it holds $\psi^c(X) + \int_0^T\bar{Z}_u\,dS_u - X \in {\cal A}$.

	(ii) \emph{Representation:} First notice that there are compact subsets $(K_n)$ of $\Omega$ (equipped with the maximum norm $||\cdot||_\infty$) such that $\Omega=\cup_{n\in \mathbb{N}}K_n$ ${\cal P}$-q.s. 
	To see this, let $P\in {\cal P}$ and $a \in A$ such that $P = P_0\circ (Y^a)^{-1}$, with $dY^a_t = a^{1/2}(Y^a_\cdot)dS_t$ $P_0$-a.s.
	Since $a$ is bounded, for every $q>4$ (independent of $P$), it holds $E_{P}\big[\big(\int_0^T|a_s|^{2}\,d\langle S\rangle_s \big)^{q/4} \big]<\infty$.
	Thus, it follows by Burkholder-Davis-Gundy and Cauchy-Schwarz' inequalities that
	\begin{align*}
		E_P\left[|Y^a_t - Y^a_s|^q \right] &= E\left[\left|\int_s^ta^{1/2}(Y^a_\cdot)\,dS_u \right|^q \right]
		\le CE_P\left[\left(\int_s^t|a(Y_\cdot^a)|\,d\langle S\rangle_u \right)^{q/2} \right]\\
		&\le CE_P\left[\left(\int_s^t|a(Y_\cdot^a)|^2\,d\langle S\rangle_u \right)^{q/4} \right]|s-t|^{q/4}\le K|s-t|^{q/4}
	\end{align*}
	for some constants $C,K\ge0$.
	Then, by \cite[Theorem A.1]{liminf}, $Y^a \in \Omega^\alpha$ for every $\alpha \in (0, 1/4-1/q)$, where $\Omega^\alpha$ is the space of functions $\omega \in \Omega$ which are $\alpha$-H\"older continuous.
	In particular, $\alpha$ can be chosen independent of $P$.
	Thus, $\Omega = \Omega^\alpha$ ${\cal P}$-q.s.
	By \cite[Corollary 3.2]{liminf}, $\Omega^\alpha=\cup_{n\in \mathbb{N}}K_n$ for some compact sets $K_n$.

	Since $0 \in {\cal A}$, for every $X \in L^\infty({\cal P})$ one has $\psi^c(X)< \infty$ and by $\psi(0)>-\infty$, it holds $\psi^c(0)\in \mathbb{R}$.
	Thus, the convex increasing function $\psi^c$ is real valued on $L^\infty(\mathcal{P})$.
	Let $(X^n)$ be an increasing sequence of bounded measurable functions such that $X^n\uparrow X$.
	By the first part of the proof, for every $n \in \mathbb{N}$ there exists $\bar{Y}^n\in G^c$ such that $\psi^c(X^n) + \bar{Y}^n - X^n \in {\cal A}$ with $\bar{Y}^n = \int_0^T\bar{Z}^n_u\,dS_u$.
	Putting $Y_t := \liminf_{n \to \infty}\int_0^t\bar{Z}^n_u\,dS_u$, $t \in [0,T]$; $\bar{Y}_t:= \limsup_{s\downarrow t, s\in \mathbb{Q}\cap [0,T]}Y_s$ for $t \in [0,T)$ and $\bar{Y}_T :=Y_T$ we use the procedure of part (i) to construct an $S$-integrable process $\bar{Z}$ such that $\int_0^t\bar{Z}_u\,dS_u\ge-c$ and $\bar{Y} \le \int \bar{Z}\,dS$.

	Let $n\in \mathbb{N}$, since $\psi^c$ is increasing and ${\cal A}$ monotone, there exists $A^n \in {\cal A}$ such that $(\lim_{n\to \infty}\psi^c(X^n))+ \bar{Y}^n - X^n= A^n $.
	Arguing as above, $(A^n)$ is bounded in $L^1({\cal P})$.
	Hence, $\lim_{n \to \infty}\psi^c(X^n) + \int_0^T\bar{Z}_u\,dS_u - X\ge\lim_{n \to \infty}\psi^c(X^n) + \bar{Y}_T - X\ge \liminf_{n \to \infty}A^n\in {\cal A}$, which implies $\lim_{n \to \infty}\psi^c(X^n) \ge \psi^c(X)$ and therefore $\lim_{n\to \infty}\psi^c(X^n) = \psi^c(X)$.
	Thus, by \cite[Theorem 1.7]{Robdual}, (see also \cite[Theorem 4.5.2]{diss} for the probabilistic version of this result) it holds
	\begin{align}
	\label{eq:rep-psitilde}
		&\psi^c(X) = \sup_{Q \in ca^+_1(\tilde\Omega)}(E_Q[X] - \psi^{c,\ast}(Q)),\quad X \in C_b({\cal P})\\
	\label{eq:rep-psitilde_Ub}
		&\psi^c(X) \le \sup_{Q \in ca^+_1(\tilde\Omega)}(E_Q[X] - \psi^{c,\ast}(Q)),\quad X \in U_b({\cal P}).
	\end{align}
	That $\psi^{c,\ast}(Q) = \phi^{c,\ast}(Q) + \rho^\ast_{{\cal A}}(-Q)$, for all $Q \in ca^+_1(\tilde\Omega)$ is a consequence of Lemma \ref{lem:convol}.
	Thus, if $\rho^*_{\cal A}(-Q)=\infty$, then $\psi^{c,*}(Q)=\infty$, so that \eqref{eq:rep_psi_c} and \eqref{eq:rep_psi_c_Ub} can be deduced from \eqref{eq:rep-psitilde} and \eqref{eq:rep-psitilde_Ub} respectively.
\end{proof}

\begin{proof}[of Theorem \ref{thm:main}]
	Since $G^n \subseteq G$ for all $n\in \mathbb{N}\setminus\{0\}$, one has $\psi(X) \le \inf_{n\ge1}\psi^n(X)$.
	Assume that the inequality is strict, that is, $\psi(X) < \inf_{n\ge1}\psi^n(X)$.
		Then, there are $m\in \mathbb{R}$ and $\varepsilon>0$ such that $\psi(X)<m<m+\varepsilon<\inf_{n\ge1}\psi^n(X)$.
		Thus, there is $x \in \mathbb{R}$ such that $\psi(X)\ge x-\varepsilon$, with $x + Y - X \in {\cal A}$ for some $Y \in G$.
		Since there is $n\in \mathbb{N}$ such that $Y \in G^n$, we have $x\ge \psi^n(X)$.
		Hence, $$m> \psi(X)\ge x-\varepsilon \ge \psi^n(X)-\varepsilon\ge \inf_{n\ge1}\psi^n(X)-\varepsilon,$$
		which is a contradiction.
	Thus, $\psi(X) = \inf_{n\ge 1}\psi^n(X)$, the sequence $(\psi^{n,*}(Q))_n$ is increasing and $\psi^*(Q) = \sup_{n\ge 1}\psi^{n,*}(Q)$ for all $Q \in ca^+_1(\tilde \Omega)$.

	Let $X \in C_b({\cal P})$.
	By Proposition \ref{pro:rep-c}, for every $n\ge 1$ it holds $\psi(X)\le \sup_{Q \in {\cal M}^{\cal A}_{{\cal P}}(S)}(E_Q[X] - \psi^{n,\ast}(Q))$.
	Thus, there is $Q^n \in ca^+_1(\tilde \Omega)$ such that
	\begin{equation*}
		\psi(X) \le E_{Q^n}[X] - \psi^{n,*}(Q^n) + \frac{1}{n}.
	\end{equation*}
	Since $X$ is bounded and $\rho^*_{\cal A}(-Q) \le \psi^{n,*}(Q)$, there is a constant $c\ge 0$ such that $Q^n \in \{Q\in ca^+_1(\tilde \Omega):\rho^*_{\cal A}(-Q) \le c\}$ for all $n$.
		Hence, there is $Q \in ca^+_1$ such that up to a subsequence, $(Q^n)$ converges to $Q$ in $\sigma(ca^+_1, C_b)$ and since $\tilde \Omega$ is closed, we have $1=\limsup_{n\to\infty}Q^n(\tilde \Omega)\le Q(\tilde \Omega)$, showing that actually, $Q \in ca^+_1(\tilde \Omega)$.
	Now, let $\varepsilon >0$ and $N \in \mathbb{N}$ be such that $\psi^{N,*}(Q)\ge \psi^*(Q) - \varepsilon$.
	Since $\psi^{N,*}$ is lower semicontinuous (with $ca^+_1$ equipped with the weak topology $\sigma(ca^+_1, C_b)$), we can choose $n\ge N$ large enough so that $\psi^{N,*}(Q^n)\ge \psi^{N,*}(Q) - \varepsilon$.
	Thus,
	\begin{equation*}
		\psi^{n,*}(Q^n) \ge \psi^{N,*}(Q^n) \ge \psi^{N,*}(Q) - \varepsilon \ge \psi^*(Q) - 2\varepsilon.	
	\end{equation*}	
	This shows that $\psi(X) \le E_{Q^n}[X] - \psi^*(Q) - 2\varepsilon + 1/n$.
	Taking the limit in $n$ and dropping $\varepsilon$ yields
	\begin{equation*}
		\psi(X) \le E_Q[X] - \psi^*(Q).
	\end{equation*}
	Since the weak duality $\psi(X)\ge \sup_{Q \in ca^+_1(\tilde \Omega)}(E_Q[X] - \psi^*(Q))$, $X \in C_b({\cal P})$ is easily obtained, this implies
	\begin{equation}
	\label{eq:rep psi Cb proof}
	 	\psi(X)= \sup_{Q \in ca^+_1(\tilde \Omega)}(E_Q[X] - \psi^*(Q))\quad X \in C_b({\cal P}).
	\end{equation}
	Let us now show that $\psi^*(Q)=\infty$ whenever $Q \notin {\cal M}^{\cal A}_{{\cal P}}(S)$.
	Since $0\in G$, we have $\psi(X) \le \rho_{\cal A}(-X)$ for every $X \in C_b$, and hence $\psi^*(Q)\ge \rho^*_{\cal A}(-Q)$ for every $Q \in ca^+_1(\tilde \Omega)$.
	Thus, if $\rho^*_{\cal A}(-Q)=\infty$, then $\psi^{*}(Q)=\infty$.
	If $S$ is not a $Q$-local martingale, then since $\text{supp}(Q)\subseteq \tilde{\Omega}$ and $\tilde \Omega$ is a subset of a $\sigma$-compact set, it follows from \cite[Remark 4.1 and Proposition 4.4]{liminf} that there is $X \in C_b$ and an $S$-integrable process $Z$ such that $X\le M^Z$ and $E_Q[X]>0$.  
	Thus, one has $\psi(xX)\le 0$ for all $x\ge 0$ and by $X\in C_b$ it holds $\psi^{\ast}(Q)\ge E_Q[xX]$ for all $x\ge 0$.
	This shows by scaling that $\psi^{\ast}(Q)=\infty$, which proves \eqref{eq:representation_psi} due to \eqref{eq:rep psi Cb proof}.

	Furthermore, it follows again by Proposition \ref{pro:rep-c} that
	\begin{equation*}
		\psi^n(X) \le \sup_{Q \in {\cal M}^{\cal A}(S)}\left(E_Q[X] - \psi^{n,\ast}(Q) \right), \quad X \in U_b({\cal P}).
	\end{equation*}
	Let $X \in U_b({\cal P})$.
	For every $n\ge1$, it holds $\psi(X) \le \sup_{Q \in {\cal M}^{\cal A}(S)}\left(E_Q[X] - \psi^{n,\ast}(Q) \right)$.
	Arguing as above, we find $Q^n, Q \in ca^+_1(\tilde \Omega)$ such that $(Q^n)$ converges to $Q$ in $\sigma(ca^+_1, C_b)$, and for every $\varepsilon>0$, up to a subsequence, $\psi(X) \le E_{Q^n}[X] - \psi^*(Q) - 2\varepsilon + 1/n$.
	Since $X$ is upper semicontinuous and bounded, taking the limit this implies $\psi(X) \le E_{Q}[X] - \psi^*(Q)$, showing that
	\begin{equation}
		\label{eq:inequal_Ub_proof}
		\psi(X) \le \sup_{Q \in {\cal M}^{\cal A}(S)}\left(E_Q[X] - \psi^{\ast}(Q) \right), \quad X \in U_b({\cal P}).
	\end{equation}
	The assumption $\psi^*(Q)= \sup_{X \in U_b}(E_Q[X]-\psi(X))$ implies that the inequality in \eqref{eq:inequal_Ub_proof} is an equality and as shown above, $\psi^*(Q) = \infty$ whenever $Q \notin {\cal M}^{\cal A}_{\cal P}(S)$.
	This concludes the proof.
\end{proof}

\subsection{Proof of Theorem \ref{thm:psiROCE}}
\label{sec:proof_thm:psiROCE}
Recall that here, ${\cal A}$  is given by 
\begin{equation*}
	{\cal A}:= \left\{X \in L^1({\cal P}): \rho(X)\le 0\right\},
\end{equation*}
the acceptance set of the robust optimized certainty equivalent defined by \eqref{eq:rob.oce}.
\begin{lemma}
\label{lem:A.UI}
	If there exist $a,b\ge0$ and $p>2$ such that the loss function $l$ satisfies the growth condition $l(x)\ge a|x|^p + b$,
	then the set ${\cal A}^-:=\{A^-: A \in {\cal A}\}$ is ${\cal P}$-uniformly integrable.
\end{lemma}
\begin{proof}
	Consider the classical OCE $\rho^P$ defined as
	\begin{equation}
	\label{eq:OCE_P}
	 	\rho^P(X):=\inf_{m \in \mathbb{R}}(E_P[l(m-X)]-m).
	\end{equation}
	By definition, we have 
	\begin{align*}
		\rho(X) \ge \sup_{P \in {\cal P}}\inf_{m \in \mathbb{R}}(E_P[l(m -X)] - m) = \sup_{P \in {\cal P}}\rho^P(X).
	\end{align*}
	Assume by contradiction that there is $\bar P \in {\cal P}$ and $\varepsilon>0$ such that it holds 
	$$\liminf_{n\to \infty}\sup_{X \in {\cal A}}E_{\bar P}[X^-1_{\{X^- \ge n\}}]\ge \varepsilon.$$
	Given $X \in {\cal A}$, put $\delta^n:= E_{\bar P}[X^-1_{\{X^-\ge n\}}]\ge \varepsilon$.
	Let $Q^n\ll \bar P$ be the measure given by $\frac{dQ^n}{d\bar P}:=X^-1_{\{X^-\ge n\}}/\delta^n$.

	Since $l^*$ satisfies the growth condition $l^*(z)\le a'|z|^q+b'$ for some $a',b'\ge 0$ and $1< q<2$ the H\"older conjugate of $p$, 
	it  holds
	\begin{align*}
		\rho(X) \ge \rho^{\bar P}(X) &\ge E_{\bar P}\left[-X\frac{X^-1_{\{X^-\ge n\}}}{\delta_n}\right] - E_{\bar P}\left[l^*\left(\frac{X^-1_{\{X^-\ge n\}}}{\delta^n}\right)\right]\\
				&\ge E_{\bar P}\left[(X^-)^21_{\{X^-\ge n\}}\right]\frac{1}{\delta^n}-a'E_{\bar P}\left[(X^-)^q1_{\{X^-\ge n\}}\right]\frac{1}{(\delta^n)^q} - b'\\
				&\ge E_{\bar P}\left[(X^-)^q1_{\{X^-\ge n\}}\right]\left(\frac{n^{2-q}}{\delta^n} - \frac{a'}{(\delta^n)^q}\right) - b'\\
				&\ge \left(E_{\bar P}\left[X^-1_{\{X^-\ge n\}}\right]\right)^q\left(\frac{n^{2-q}}{\delta^n} - \frac{a'}{(\delta^n)^q}\right) - b',	
	\end{align*} 
	where the last inequality follows by Jensen's inequality with $n$ large enough.
	Since $q<2$, the last term converges to infinity, a contradiction.
\end{proof}
Recall the conjugate function $\rho^*$ defined in Section \ref{sec:main results} as
\begin{equation*}
	\rho^*(-Q):= \sup_{X \in C_b}(E_Q[-X] - \rho(X)).
\end{equation*}
\begin{lemma}
\label{lem:A.sublevel}
	If ${\cal P}$ is $\sigma(ca^+_1,C_b)$-compact, then it holds
	\begin{equation}
	\label{eq:rho.start.inf}
		\rho^*(-Q) = \inf_{P \in {\cal P}}E_P\left[l^*\left(\frac{dQ}{dP} \right) \right] \quad \text{for all } Q \in ca^+_1
	\end{equation} 
	and the sublevel sets $\{Q \in ca^+_1: \rho^*(-Q)\le c\}$, $c\ge 0$ are $\sigma(ca^+_1, C_b)$-compact.
\end{lemma}
\begin{proof}
	Let us first prove \eqref{eq:rho.start.inf}. 	
	Since for each $X \in C_b$ the function $P\mapsto \rho^P(X)$ is concave and $\sigma(ca^+_1, C_b)$-upper semicontinuous, it follows by weak compactness of ${\cal P}$ and \citet[Theorem 2]{fan53} that 
	\begin{equation}
	\label{eq:infsup1}
	 	\rho^\ast(-Q) = \inf_{P \in {\cal P}}\sup_{X \in C_b}(E_Q[-X]-\rho^P(X)).
	\end{equation}
	Let $L^\infty(P, {\cal F}^S_T)$ be the space of $P$-essentially bounded and ${\cal F}^S_T$-measurable random variables.
	We claim that 
	$$\sup_{X \in C_b}(E_Q[-X]-\rho^P(X)) = \sup_{X \in L^\infty(P,{\cal F}^S_T)}(E_Q[-X]-\rho^P(X)) = E_P[l^*(dQ/dP)]$$ for every Borel measure $Q \ll P$.
	The second equality of the claim follows by \cite{Ben-Teb}.
	To prove the first one, let $\varepsilon>0$ and $X \in L^\infty(P,{\cal F}^S_T)$ be such that $$\sup_{X \in L^\infty(P,{\cal F}^S_T)}(E_Q[-X]-\rho^P(X))\le E_Q[-X]-\rho^P(X) + \varepsilon.$$
	It follows by Lusin's and Tietze's theorems that there is a sequence $(X^n)$ of continuous functions converging $P$-a.s. to $X$, see for instance \cite[Theorem 1]{Wis94} for details.
	In addition, the sequence $(X^n)$ can be chosen bounded. 
	Since $l^*(x)/x\to +\infty$ as $|x|$ goes to infinity, it follows that for each $c\ge0$ the set $\{dQ/dP: E_P[l^*(dQ/dP)]\le c\}$ is $\sigma(L^1(P,{\cal F}^S_T),L^\infty(P,{\cal F}^S_T))$-compact.
	Hence, by the representation
	\begin{equation*}
	 \label{eq:rep_rho-P}
	 	\rho^P(X) = \sup_{Q\ll P}(E_Q[-X]-E_P[l^*(dQ/dP)])\quad X \in L^\infty(P, {\cal F}^S_T),
	\end{equation*}
	see e.g. \cite[Theorem 4.2]{Ben-Teb}, and the Jouini-Schachermayer-Touzi theorem, (see \cite[Theorem 2.4]{jou-sch-tou}) one has $\lim_{n\to \infty}\rho^P(X^n)=\rho^P(X)$.
	Therefore, $E_Q[-X]-\rho^P(X) = \lim_{n\to \infty}E_Q[-X^n] - \rho^P(X^n)$, which proves the claim.
	In combination with \eqref{eq:infsup1}, we obtain \eqref{eq:rho.start.inf}.

	Next, let us prove compactness of the sublevel sets.
	As a consequence of Prokhorov's theorem the set ${\cal P}$ is tight.
    That is, there exists a family $(K_n)$ of compact subsets of $\Omega$ 
    such that $\sup_{P\in {\cal P}}P(K_{n}^c)\to 0$.
    Let $Q\in ca^1_+$ satisfy $\rho^*(-Q)\le C$.
    There is $P \in {\cal P}$ such that $Q\ll P$ and $E_P[l^*(\frac{dQ}{dP})]\le C+1$.
    Therefore for all $m>0$, by Young's inequality one has
    \begin{align*}
       m \frac{dQ}{dP}1_{K_n^c} \le l(m1_{K_n^c}) + l^*(\frac{dQ}{dP})
    \end{align*}
    and using $l(0) = 0$ one gets
    \begin{align*}
        Q(K_n^c)
                  &\le\frac{1}{m}\left( l(m)P(K^c_n)+ E_P\left[l^*(\frac{dQ}{dP})\right]\right)
                 \le \frac{l(m)}{m}\sup_{P\in {\cal P}}P(K_n^c) + \frac{(C+1)}{m}.
    \end{align*}
    Thus, $\sup_{\{Q:\rho^*(-Q)\le C\}}Q(K^c_n)\to 0$ as $n\to \infty$ showing by Prokhorov's theorem that the sublevel set is weakly relatively compact.
    Since $\rho^*$ is lower semicontinuous, the sublevel sets are $\sigma(ca_1^+,C_b)$-closed.
    This completes the argument.
\end{proof}
\begin{lemma}
\label{lem:Q-Y}
	For every local martingale measure $Q$ of $S$ such that $Q\ll P$ for some $P \in {\cal P}$, and every $Y \in G$, it holds $E_Q[Y]\le 0$.
\end{lemma}
\begin{remark}
	Note that in the above lemma, $Y\in G$ is not necessarily a $Q$-stochastic integral.
\end{remark}
\begin{proof}[of Lemma \ref{lem:Q-Y}]
	Let $Y\in G$ and $Q$ be a local martingale measure for $S$.
	Let $P\in {\cal P}$ be such that $Q\ll P$, let $c>0$ and $Z$ an $S$-integrable process such that $Y_t:=\int_0^tZ_u\,dS_u\ge -c$.
	Recall that a process $H:[0,T]\times \Omega\to \mathbb{R}$ is called simple if it is of the form
	\begin{equation*}
	 	H_t(\omega)= \sum_{i=1}^Nh_i(\omega)1_{(\tau_i(\omega), \tau_{i+1}(\omega)]}
	 \end{equation*}
	 where $N\in \mathbb{N}$, $0\le \tau_1\le \cdots\le \tau_{N+1}\le T$ are $\mathbb{F}$-stopping times and $h_i$ are ${\cal F}_{\tau_i}$-measurable bounded functions.

	Let us first assume that $E_P[\int_0^T|Z_u|^2\,du]<\infty$.
	Then, there is a sequence $(Z^n)$ of simple processes such that $\int_0^T Z^n_u\,dS_u \to \int_0^TZ_u\,dS_u$ in $L^2(P)$.
	Fix $\varepsilon>0$, and define the sequence of stopping times
	\begin{equation*}
		\tau^n:= \inf\left\{t>0: \int_0^tZ^n_u\,dS_u \le -c-\varepsilon\right\}\wedge T,
	\end{equation*}
	with the convention $\inf\emptyset:= +\infty$.
	Further put $\tilde{Z}^n:= Z^n1_{[0,\tau^n]}$.
	By definition, $\int_0^t\tilde{Z}_u^n\,dS_u = \int_0^{t\wedge \tau^n}Z^n_u\,dS_u\ge - c - \varepsilon$.
	For almost all $\omega \in \Omega$, there is $N \in \mathbb{N}$ such that if $n \ge N$, then $\int_0^tZ^n_u\,dS_u(\omega)\ge \int_0^tZ_u\,dS_u(\omega)-\varepsilon \ge -c-\varepsilon$.
	Thus, $\tau^n(\omega) =T$, i.e. $\tau^n\uparrow T$ $P$-a.s.
	Hence, $\int_0^T\tilde{Z}^n_u\,dS_u = \int_0^{\tau^n}Z^n_u\,dS_u$ converges to $\int_0^TZ_u\,dS_u$ $P$-a.s. and $Q$-a.s.
	In addition, since $\tilde{Z}^n$ is a simple process, $\int \tilde{Z}^n\,dS$ is a $Q$-martingale, so that $E_Q[\int_0^T\tilde{Z}^n_u\,dS_u] = 0$.
	Therefore, it follows from Fatou's lemma that 
	\begin{equation*}
		0= \liminf_{n \to \infty}E_Q\left[\int_0^T\tilde{Z}^n_u\,dS_u\right] \ge E_Q\left[ \int_0^TZ_u\,dS_u\right].
	\end{equation*}
	
	In the general case, let $\sigma^k$ be a localizing sequence such that $\int_0^{\sigma^k\wedge\cdot}Z\,dS$ is a square integrable $P$-martingale.
	Put $Z^k:= Z1_{[0,\sigma^k]}$.
	One has $\int_0^tZ^k_u\,dS_u\ge -c$ for all $k$.
	By the first part of the proof, for each $k$, it holds $ E_Q[\int_0^TZ^k_u\,dS_u]\le 0$.
	Taking the limit in $k$, it follows by Fatou's lemma that $E_Q[\int_0^TZ_u\,dS_u]\le 0$.
\end{proof}

\begin{proof}[of Theorem \ref{thm:psiROCE}]
	It is clear that the set ${\cal A}$ contains $0$, is convex and monotone.
	Moreover, in view of Theorem \ref{thm:main}, Lemmas \ref{lem:A.UI} and \ref{lem:A.sublevel}, the representation can be obtained if we show that $X:=\liminf_{n\to \infty}X^n\in{\cal A}$ for every sequence $(X^n) $ in $ {\cal A}$ that is bounded in $L^1({\cal P})$ and that
	\begin{equation}
	\label{eq:psi_ub_proof}
		\psi^\ast(Q)= \psi^\ast_{U_b}(Q):=\sup_{X \in U_b}(E_Q[X]-\psi(X)).
	\end{equation}

	For every $n \in \mathbb{N}$, there is 
	$m^n\in \mathbb{R}$ such that
	$$E_P[l(m^n -X^n)] - m^n-1/n\le 0 \quad \text{for all }  P \in {\cal P}.$$
	Condition (CIB) ensures that $l(x)\ge bx + c$ and $l(x) \ge b'x + c$ for all $x \in \mathbb{R}$ for some $b>1>b'$ and $c \in \mathbb{R}$.
	Since $(X^n)$ is bounded in $L^1({\cal P})$, this shows that $(m^n)$ is bounded.
	Thus, there is $m$ such that $(m^n)$ converges to $m$ after passing to a subsequence.
	Hence, it follows from Fatou's lemma and continuity of $l$ that $E_P[l(m - X)] - m \le 0$.
	Since this holds for every $P\in {\cal P}$, it follows $\rho(X) \le 0$, i.e. $X \in {\cal A}$.

	Let us now prove \eqref{eq:psi_ub_proof}.
	It follows from Theorem \ref{thm:main} and Lemma \ref{lem:A.sublevel} that
	\begin{equation}
	\label{eq:firstinequal_es}
		\psi^\ast(Q)= \phi^\ast(Q) +\inf_{P\in {\cal P}}E_P[l^*(dQ/dP)].
	\end{equation}	
	Observe that by definition, $\psi^* \le \psi^*_{U_b}$.
	Let $Q \in ca^+_1(\tilde \Omega)$.
	Assuming that $\inf_{P \in {\cal P}}E_P[l^*(dQ/dP )]=\infty$, then by Lemma \ref{lem:A.sublevel}, $\psi^*(Q)\ge \rho^*(-Q)= \infty$.
	If $S$ is not a $Q$-local martingale, then since $\text{supp}(Q)\subseteq \tilde{\Omega}$ and $\tilde \Omega$ is a subset of a $\sigma$-compact set, it follows from \cite[Remark 4.1 and Proposition 4.4]{liminf} that there is $X \in C_b$ and an $S$-integrable process $Z$ such that $X\le M^Z$ and $E_Q[X]>0$.  
	Thus, one has $\phi(xX)\le 0$ for all $x\ge 0$ and by $X\in C_b$ it holds $\psi^*(Q)\ge \phi^*(Q)\ge E_Q[xX]$ for all $x\ge 0$.
	This shows by scaling that $\psi^*(Q) = \phi^*(Q) = \infty$.
	Thus, $\infty=\psi^*(Q) \le \psi^*_{U_b}(Q)$ for all $Q \notin {\cal M}^l_{\cal P}(S)$.

	On the other hand, it can be checked that $\rho$ satisfies the weak duality 
	$$\rho(X) \ge \sup_{Q}(E_Q[-X] - \inf_{P \in {\cal P}}E_P[l^*(dQ/dP)])\quad \text{for all }X \in L^1({\cal P}).$$
	Let $Q \in {\cal M}^l_{\cal P}(S)$, and $X \in U_b$ such that $m + Y -X \in {\cal A}$ for some $m\in \mathbb{R}$ and $Y \in G$.
	It holds $E_Q[-m-Y+X]-\rho^*(-Q)\le 0$ and there is $P \in {\cal P}$ such that $Q\ll P$.
	By Lemma \ref{lem:Q-Y}, we have $E_Q[Y]\le 0$, i.e., $E_Q[X]-m\le \rho^*(-Q)$.
	This implies $\psi^*_{U_b}(Q)\le\rho^*(-Q)\le \psi^*(Q)$.
	Therefore, $ \psi^*=\psi^*_{U_b}$.

	Finally, recall that $\psi^*(Q) = \phi^*(Q) + \rho^*(-Q)$ for every $Q$ and $\phi^*(Q) = 0$ for $Q \in {\cal M}^l_{\cal P}(S)$.
	This concludes the proof. 
\end{proof}

\begin{appendix}
\section{Separable class of diffusion coefficients}
In this appendix we define the classes of diffusion coefficients we consider.
The definitions below as well as the proposition are taken from \citet{STZ10}.
Let $\bar{\cal P}_W$ be the set of local martingale measures $P$ of $S$ such that $P$-a.s., $\langle S\rangle_t$ is absolutely continuous in $t$ and $\hat{a}$ takes values in $\mathbb{S}^{>0}_d$, put
\begin{equation*}
	\bar{A} := \left\{a:\mathbb{R}_+\to \mathbb{S}^{>0}_d,\,\, \mathbb{F}\text{-progressively measurable and } \int_0^t|a_s|\,ds<\infty \text{ for all } t\ge 0 \right\}
\end{equation*}
and for each $P \in \bar{\cal P}_W$, $\bar{A}_W(P):= \{a \in \bar{A}: a = \hat{a}\, P\text{-a.s.}\}$.
Denote by $\bar{A}_W$ the set $\bar{A}_W:=\cup_{P \in \bar{\cal P}_W}\bar{A}_W(P)$ and by $A_W$ the set of elements of $\bar{A}_W$ such that the SDE \eqref{eq:sde} has weak uniqueness.
\begin{definition}
\label{def:generating}
	A subset $A_0$ of $A_W$ is called a generating class of diffusion coefficients if
	\begin{itemize}
	\item[(i)] $A_0$ satisfies the concatenation property: $a1_{[0,t)} + b1_{[t,\infty)} \in A_0$ for all $a,b \in A_0$ and $t\ge0$.
	\item[(ii)] $A_0$ has constant disagreement times: for all $a,b \in A_0$, $\theta^{a,b}$ is a constant, with $\theta^{a,b}:\inf\{t\ge 0:\int_0^ta_s\,ds \neq \int_0^tb_s\,ds\}$.
	\end{itemize}
\end{definition}
\begin{definition}
\label{def:separable}
	A set $A$ is a separable class of diffusion coefficients generated by $A_0$ if $A_0\subseteq A_W$ is a generating class of diffusion coefficients and $A$ consists of all processes of the form
	\begin{equation*}
		a = \sum_{n = 0}^\infty\sum_{i=1}^\infty a^n_i1_{E^n_i}1_{[\tau_n, \tau_{n+1})},
	\end{equation*}
	where $(a^n_i)_{i,n}\subseteq A_0$, $(\tau_n)_n$ is an increasing sequence of $\mathbb{F}$-stopping times valued in $\mathbb{R}_+\cup {\{\infty\}}$ with $\tau_0 = 0$ and
	\begin{itemize}
	\item[(i)] $\inf\{n:\tau_n=\infty\}<\infty$, $\tau_n<\tau_{n+1}$ whenever $\tau_n<\infty$, and each $\tau_n$ takes at most countably many values.
	\item[(ii)] for each $n$, $\{E^n_i, i\ge 1\}\subseteq {\cal F}_{\tau_n}$ forms a partition of $\Omega$.
	\end{itemize}
\end{definition}
\begin{proposition}
	Let $ A$ be a separable class of diffusion coefficients generated by $A_0$.
	Then, $A\subseteq A_W$ and if for all $a \in A_0$ $P^a$ satisfies the martingale representation property, then for all $a \in A$ $P^a$ satisfies the martingale representation property as well.
\end{proposition}
\end{appendix}


\begin{thebibliography}{35}
\providecommand{\natexlab}[1]{#1}
\providecommand{\url}[1]{\texttt{#1}}
\expandafter\ifx\csname urlstyle\endcsname\relax
  \providecommand{\doi}[1]{doi: #1}\else
  \providecommand{\doi}{doi: \begingroup \urlstyle{rm}\Url}\fi

\bibitem[Acciaio et~al.(2016)Acciaio, Beiglb\"ock, Penkner, and
  Schachermayer]{Acciaio2016}
B.~Acciaio, M.~Beiglb\"ock, F.~Penkner, and W.~Schachermayer.
\newblock A model-free version of the fundamental theorem of asset pricing and
  the super-replication theorem.
\newblock \emph{Math. Finance}, 26\penalty0 (2):\penalty0 233--251, 2016.

\bibitem[Arai(2010)]{Arai10}
T.~Arai.
\newblock Convex risk measures on orlicz spaces: inf-convolution and shortfall.
\newblock \emph{Math. Finan. Econ.}, 3:\penalty0 73--88, 2010.

\bibitem[Bartl et~al.(2017)Bartl, Drapeau, and Tangpi]{Roboptim}
D.~Bartl, S.~Drapeau, and L.~Tangpi.
\newblock Computational aspects of robust optimized certainty equivalents.
\newblock \emph{Preprint arXiv:1706.10186}, 2017.

\bibitem[Bartl. et~al.(2017)Bartl., Kupper, Pr\"omel, and Tangpi]{liminf}
D.~Bartl., M.~Kupper, D.~J. Pr\"omel, and L.~Tangpi.
\newblock Duality for pathwise superhedging in continuous time.
\newblock \emph{Preprint arXiv:1705.02933}, 2017.

\bibitem[Bartl et~al.(2017)Bartl, Neufeld, and Kupper]{Bar-Neu-Kup17}
D.~Bartl, A.~Neufeld, and M.~Kupper.
\newblock Pathwise superhedging on prediction sets.
\newblock \emph{Preprint arXiv:1711.02764}, 2017.

\bibitem[Becherer and Kentia(2017)]{Bec-Ken17}
D.~Becherer and K.~Kentia.
\newblock Good deal hedging and valuation under combined uncertainty about
  drift and volatility.
\newblock \emph{Probability, Uncertainty and Quantitative Risk}, 2\penalty0
  (13), 2017.

\bibitem[Beiglb\"ock et~al.(2013)Beiglb\"ock, Henry-Labord\`ere, and
  Penkner]{bei-hl-pen}
M.~Beiglb\"ock, P.~Henry-Labord\`ere, and F.~Penkner.
\newblock Model-independent bounds for option prices -- a mass transport
  approach.
\newblock \emph{Finance Stoch.}, 17\penalty0 (3):\penalty0 477--501, 2013.

\bibitem[Ben-Tal and Taboulle(2007)]{Ben-Teb}
A.~Ben-Tal and M.~Taboulle.
\newblock An old-new concept of convex risk measures: The optimized certainty
  equivalent.
\newblock \emph{Math. Finance}, 17:\penalty0 449--476, 2007.

\bibitem[Bion-Nadal and {Di Nunno}(2013)]{Nadal}
J.~Bion-Nadal and G.~{Di Nunno}.
\newblock {Dynamic no-good-deal pricing measures and extension theorems for
  linear operators on $L^\infty$}.
\newblock \emph{Finance Stoch.}, 17\penalty0 (3):\penalty0 587--613, 2013.

\bibitem[Bion-Nadal and Kervarec(2012)]{Bion-Kervarec}
J.~Bion-Nadal and M.~Kervarec.
\newblock {Risk Measuring under Model Uncertainty}.
\newblock \emph{Ann. Appl. Probab.}, 22\penalty0 (1):\penalty0 213--238, 2012.

\bibitem[Burzoni(2016)]{Burz16}
M.~Burzoni.
\newblock Arbitrage and hedging in model independent markets with frictions.
\newblock \emph{SIAM J. Financial Math.}, 7\penalty0 (1):\penalty0 812--844,
  2016.

\bibitem[Burzoni et~al.(2017)Burzoni, Frittelli, and Maggis]{Bur-Fri-Mag17}
M.~Burzoni, M.~Frittelli, and M.~Maggis.
\newblock Model-free superhedging duality.
\newblock \emph{Ann. Appl. Probab.}, 27\penalty0 (3):\penalty0 1452--1477,
  2017.

\bibitem[Cheridito et~al.(2015)Cheridito, Kupper, and Tangpi]{Robdual}
P.~Cheridito, M.~Kupper, and L.~Tangpi.
\newblock Representation of increasing convex functionals with countably
  additive measures.
\newblock Preprint, 2015.

\bibitem[Cheridito et~al.(2017)Cheridito, Kupper, and Tangpi]{robhedging}
P.~Cheridito, M.~Kupper, and L.~Tangpi.
\newblock Duality formulas for robust pricing and hedging in discrete time.
\newblock \emph{SIAM J. Financial Math.}, 8\penalty0 (1):\penalty0 738--765,
  2017.

\bibitem[Delbaen and Schachermayer(1994)]{DS94}
F.~Delbaen and W.~Schachermayer.
\newblock A general version of the fundamental theorem of asset pricing.
\newblock \emph{Math. Ann.}, 300\penalty0 (3):\penalty0 463--520, 1994.

\bibitem[Dellacherie and Meyer(1982)]{Dellacherie1982}
C.~Dellacherie and P.-A. Meyer.
\newblock \emph{Probabilities and Potential. B}, volume~72 of
  \emph{North-Holland Mathematics Studies}.
\newblock North-Holland Publishing Co., Amsterdam, 1982.
\newblock Theory of martingales, Translated from the French by J. P. Wilson.

\bibitem[Denis and Martini(2006)]{denis06}
L.~Denis and C.~Martini.
\newblock A theoretical framework for the pricing of contingent claims in the
  presence of model uncertainty.
\newblock \emph{Ann. Appl. Probab.}, 16:\penalty0 827--852, 2006.

\bibitem[Dolinsky and Soner(2014)]{Dolinsky2014}
Y.~Dolinsky and H.~M. Soner.
\newblock Martingale optimal transport and robust hedging in continuous time.
\newblock \emph{Probab. Theory Related Fields}, 160\penalty0 (1-2):\penalty0
  391--427, 2014.

\bibitem[Fan(1953)]{fan53}
K.~Fan.
\newblock Minimax theorems.
\newblock \emph{Proc. Nat. Acad. Sci. U.S.A}, 39:\penalty0 42--47, 1953.

\bibitem[F\"{o}llmer and Leukert(1999)]{Foel-Leu}
H.~F\"{o}llmer and P.~Leukert.
\newblock Quantile hedging.
\newblock \emph{Finance Stoch.}, 3\penalty0 (3):\penalty0 251--273, 1999.

\bibitem[F\"{o}llmer and Leukert(2000)]{Foel-Leu2000}
H.~F\"{o}llmer and P.~Leukert.
\newblock Efficient hedging: Cost versus shortfall risk.
\newblock \emph{Finance Stoch.}, 4:\penalty0 117--146, 2000.

\bibitem[Hou and Ob{\l}\'{o}j(2018)]{Hou2015}
Z.~Hou and J.~Ob{\l}\'{o}j.
\newblock On robust pricing-hedging duality in continuous time.
\newblock \emph{Finance Stoch.}, 22\penalty0 (3):\penalty0 511--567, 2018.

\bibitem[Jouini et~al.(2006)Jouini, Schachermayer, and Touzi]{jou-sch-tou}
E.~Jouini, W.~Schachermayer, and N.~Touzi.
\newblock Law invariant risk measures have the fatou property.
\newblock \emph{Advances in Mathematical Economies}, 9:\penalty0 49--71, 2006.

\bibitem[Kaina and R\"uschendorf(2009)]{Kaina-Ruesch}
M.~Kaina and L.~R\"uschendorf.
\newblock On convex risk measures on $l^p$-spaces.
\newblock \emph{Math. Meth. Oper. Res.}, 69\penalty0 (3):\penalty0 475--495,
  2009.

\bibitem[Karandikar(1983)]{karandikar83}
R.~L. Karandikar.
\newblock On quadratic variation process of a continuous martingale.
\newblock \emph{Illinois J. Math.}, 27:\penalty0 178--181, 1983.

\bibitem[Karatzas and Shreve(2004)]{karatzas01}
I.~Karatzas and S.~E. Shreve.
\newblock \emph{Brownian Motion and Stochastic Calculus (Graduate Texts in
  Mathematics)}.
\newblock Springer, August 2004.

\bibitem[Kramkov and Schachermayer(1999)]{Kra-Sch}
D.~Kramkov and W.~Schachermayer.
\newblock {The Asymptotic Elasticity of Utility Functions and Optimal
  Investment in Incomplete Market}.
\newblock \emph{Ann. Appl. Probab.}, 9\penalty0 (3):\penalty0 904--950, 1999.

\bibitem[Neufeld and Nutz(2013)]{Neu-Nutz}
A.~Neufeld and M.~Nutz.
\newblock {Superreplication under Volatility Uncertainty for Measurable
  Claims}.
\newblock \emph{Electron. J. Probab.}, 18\penalty0 (48):\penalty0 1--14, 2013.

\bibitem[Peng(2010)]{peng_2010}
S.~Peng.
\newblock Nonlinear expectations and stochastic calculus under uncertainty.
\newblock arXiv Preprint 1002.4546., 2010.

\bibitem[Rudloff(2007)]{Rud07}
B.~Rudloff.
\newblock Convex hedging in incomplete markets.
\newblock \emph{Appl. Math. Finance}, 14\penalty0 (5):\penalty0 437--452, 2007.

\bibitem[Soner et~al.(2011{\natexlab{a}})Soner, Touzi, and Zhang]{STZ2}
H.~M. Soner, N.~Touzi, and J.~Zhang.
\newblock Martingale representation theorem for the {$G$}-expectation.
\newblock \emph{Stoch. Proc. Appl.}, 121\penalty0 (2):\penalty0 265--287,
  2011{\natexlab{a}}.

\bibitem[Soner et~al.(2013)Soner, Touzi, and Zhang]{STZ1}
H.~M. Soner, N.~Touzi, and J.~Zhang.
\newblock Dual formulation of second order target problems.
\newblock \emph{Ann. Appl. Probab.}, 23\penalty0 (1):\penalty0 308--347, 2013.

\bibitem[Soner et~al.(2011{\natexlab{b}})Soner, Touzi, and Zhang]{STZ10}
M.~H. Soner, N.~Touzi, and J.~Zhang.
\newblock Quasi-sure stochastic analysis through aggregation.
\newblock \emph{Electron. J. Probab.}, 16\penalty0 (67), 1844-1879
  2011{\natexlab{b}}.

\bibitem[Tangpi(2015)]{diss}
L.~Tangpi.
\newblock \emph{{Dual Representation of Convex Increasing Functionals with
  Applications to Finance}}.
\newblock PhD thesis, University of Konstanz, 2015.

\bibitem[Wisniewski(1994)]{Wis94}
A.~Wisniewski.
\newblock The structure of measurable mappings on metric spaces.
\newblock \emph{Proc. A.M.S.}, 122\penalty0 (1):\penalty0 147--150, 1994.

\end{thebibliography}

 \vspace{1cm}

\noindent Ludovic Tangpi: Department of Operations Research and Financial Engineering, Princeton University, Princeton, 08540, NJ;  USA.\\
{\small\textit{E-mail address:} ludovic.tangpi@princeton.edu}.

\end{document}